\documentclass[english]{cccconf}

\usepackage[english]{babel}
\usepackage{theorem}
\theoremheaderfont{\itshape\bfseries}
{\theorembodyfont{\itshape}
\newtheorem{assumption}{\textbf{Assumption}}
\newtheorem{lemma}{\textbf{Lemma}}
\newtheorem{definition}{\textbf{Definition}}
\newtheorem{theorem}{\textbf{Theorem}}

\newtheorem{remark}{\textbf{Remark}}

\newtheorem{problem}{\textbf{Problem}}
}
\usepackage{graphicx}
\usepackage{newtxtext,newtxmath}
\usepackage{amsmath}
\usepackage{amsfonts}
\usepackage{amssymb}
\usepackage{mathrsfs}
\usepackage{xcolor}
\usepackage{graphicx}

\newcommand{\T}{^{\mbox{\tiny T}}}
\newcommand{\R}{\mathbb{R}}
\newcommand{\C}{\mathbb{C}}
\newcommand{\N}{\mathbb{N}}
\newcommand{\eps}{\varepsilon}
\let\leq\leqslant
\let\geq\geqslant

\newenvironment{proof}[1][Proof]%
{\par\noindent\textit{#1:\ }}%
{\hspace*{\fill} \rule{6pt}{6pt}}
\newenvironment{proof*}[1][Proof]%
{\par\noindent\textit{#1:\ }}{}

\DeclareMathOperator{\diag}{diag}

\DeclareMathOperator{\re}{Re}

\usepackage{tikz}
\usetikzlibrary{shapes,calc,arrows,patterns,decorations.pathmorphing
	,decorations.markings}
\usetikzlibrary{arrows.meta}

\newenvironment{system}[1]%
{\setlength{\arraycolsep}{0.5mm}\left\{ \; \begin{array}{#1}}%
    {\end{array} \right.}
\newenvironment{system*}[1]%
{\setlength{\arraycolsep}{0.5mm} \begin{array}{#1}}%
  {\end{array}}

\begin{document}

\title{Regulated State Synchronization for Discrete-Time Homogeneous Networks of Non-Introspective Agents in Presence of Unknown Non-Uniform Input Delays: A Scale-Free Protocol Design}
\author{Zhenwei Liu\aref{neu}, Donya Nojavanzadeh\aref{wsu}, Dmitri Saberi\aref{stu}, Ali Saberi\aref{wsu},
	Anton A. Stoorvogel\aref{ut}}

\affiliation[neu]{College of Information Science and
	Engineering, Northeastern University, Shenyang 110819, China
	\email{liuzhenwei@ise.neu.edu.cn}}
\affiliation[wsu]{School of Electrical Engineering and Computer
	Science, Washington State University, Pullman, WA 99164, USA
	\email{donya.nojavanzadeh@wsu.edu; saberi@eecs.wsu.edu}}
\affiliation[stu]{Stanford University, Stanford, CA 94305, USA\email{dsaberi@stanford.edu}}
\affiliation[ut]{Department of Electrical Engineering,
	Mathematics and Computer Science, University of Twente, Enschede, The Netherlands
	\email{A.A.Stoorvogel@utwente.nl}}

\maketitle

\begin{abstract}
   This paper studies regulated state synchronization of discrete-time homogeneous networks of non-introspective agents in presence of unknown non-uniform input delays. A scale free protocol is designed based on additional information exchange, which does not need any knowledge of the directed network topology and the spectrum of the associated Laplacian matrix. The proposed protocol is scalable and achieves state synchronization for any arbitrary number of agents. Meanwhile, an upper bound for the input delay tolerance is obtained, which explicitly depends on the agent dynamics.
  \end{abstract}

	
	\footnotetext{This work is supported by Nature Science
		Foundation of Liaoning Province under Grant 2019-MS-116.}
\section{Introduction}
Cooperative control of multi-agent systems (MAS) is an
active research topic because of its widespread application in different areas such as sensor networks, automotive vehicle
control, satellite or robot formation, power distribution
systems and so on. See for instance the books \cite{ren-book} and \cite{wu-book} or the survey paper \cite{saber-murray3}.

We identify two classes of MAS: homogeneous and heterogeneous. 
State synchronization inherently requires homogeneous networks (i.e. networks with identical agent models). Therefore, in this paper, our focus is on homogeneous networks of MASs. State synchronization based on diffusive full-state coupling has been considered in the literature where the agent dynamics progress from single- and double-integrator
(e.g.  \cite{saber-murray2}, \cite{ren}, \cite{ren-beard}) to more
general dynamics (e.g. \cite{scardovi-sepulchre}, \cite{tuna1},
\cite{wieland-kim-allgower}). State synchronization based on
diffusive partial-state coupling has also been considered, including static design (\cite{liu-zhang-saberi-stoorvogel-auto} and \cite{liu-zhang-saberi-stoorvogel-ejc}), dynamic design (\cite{kim-shim-back-seo}, \cite{seo-back-kim-shim-iet}, \cite{seo-shim-back}, \cite{su-huang-tac},
\cite{tuna3}), and the design with additional communication (\cite{chowdhury-khalil} and \cite{scardovi-sepulchre}). 

A common assumption, especially for heterogeneous MAS, is that agents are introspective; that is, agents possess some knowledge about their own states. So far there exist many results about this type of agents, see for instance \cite{chen-auto2019, kim-shim-seo, li-soh-xie-lewis-TAC2019,  qian-liu-feng-TAC2019,yang-saberi-stoorvogel-grip-journal}.
On the other hand, for non-introspective agents, regulated output synchronization for a heterogeneous network is studied in
\cite{peymani-grip-saberi,peymani-grip-saberi-fossen}.  Other designs can also be found, such as an
internal model principle based design
\cite{wieland-sepulchre-allgower}, distributed high-gain observer
based design \cite{grip-yang-saberi-stoorvogel-automatica}, 
low-and-high gain based, purely distributed, linear time invariant protocol design \cite{grip-saberi-stoorvogel}.

In practical applications, the network dynamics are not perfect and may be subject to delays. Time delays may afflict systems performance or even lead to instability. As discussed in \cite{cao-yu-ren-chen}, two kinds of delay have been considered in the literature: input delay and communication delay. Input delay is the processing time to execute an input for each agent whereas communication delay can be considered as the time for transmitting information from origin agent to its destination.
Some researches have been done in the case of communication delay \cite{chopra-tac,chopra-spong-cdc06,ghabcheloo,klotz-obuz-kan,lin-jia-auto, munz-papachristodoulou-allgower,munz-papachristodoulou-allgower2,stoorvogel-saberi-cdc16,tian-liu,xiao-wang-tac}.  In the case of input delays, many efforts have been done (see \cite{bliman,lin-jia,saber-murray2, tian-liu, xiao-wang}) where they are mostly restricted to simple agent models such as first and second-order dynamics for both linear and nonlinear agent dynamics. \cite{wang-aut,wang-saberi-stoorvogel-grip-yang} studied state synchronization problems in the presence of unknown uniform constant input delay for both continuous- and discrete-time networks with higher-order linear agents. Recently, \cite{zhang-saberi-stoorvogel-continues-discrete} has studied synchronization in homogeneous networks of both continuous- or discrete-time agents with unknown non-uniform constant input delays.

A common characteristic in all of the aforementioned works either with input delay or communication delay is that the proposed protocols require some knowledge of communication networks that is the spectrum of associated Laplacian matrix and obviously the number of agents. In contrast, by virtue of an additional information exchange for MAS with both full- and partial-state coupling, we design and present protocols with the following distinctive characteristics: 
\begin{itemize}
	\item  The design is independent of information about communication networks. That is to say, the dynamical protocol can work for any communication network such that all of its nodes have path to the exosystem. 
	\item  The dynamic protocols are designed for networks with unknown non-uniform input delays where the admissible upper bound on delays only depends on agent model and does not depend on communication network and the number of agents.
	\item The proposed protocols are scale-free: they achieve regulated state synchronization for any MAS with any number of agents, any admissible non-uniform input delays,
	and any communication network.
\end{itemize}

\subsection*{Notations and definitions}

Given a matrix $A\in \mathbb{R}^{m\times n}$, $A\T$ denotes its
conjugate transpose and $\|A\|$ is the induced 2-norm while $\sigma_{\min}(A)$ denotes the smallest singular value of A. Let $j$ indicate $\sqrt{-1}$. A square matrix
$A$ is said to be Schur stable if all its eigenvalues are in the closed unit disc. We denote by
$\diag\{A_1,\ldots, A_N \}$, a block-diagonal matrix with
$A_1,\ldots,A_N$ as the diagonal elements. $A\otimes B$ depicts the
Kronecker product between $A$ and $B$. $I_n$ denotes the
$n$-dimensional identity matrix and $0_n$ denotes $n\times n$ zero
matrix; sometimes we drop the subscript if the dimension is clear from
the context. Moreover, $\ell_\infty^n(K)$ denote the Banach space of finite sequences $\{y_1,\hdots,y_K\}\subset \mathbb{C}^n$ with norm $\|.\|_\infty=\max_{i}\{\|y_i\|\}$.

To describe the information flow among the agents we associate a \emph{weighted graph} $\mathcal{G}$ to the communication network. The weighted graph $\mathcal{G}$ is defined by a triple
$(\mathcal{V}, \mathcal{E}, \mathcal{A})$ where
$\mathcal{V}=\{1,\ldots, N\}$ is a node set, $\mathcal{E}$ is a set of
pairs of nodes indicating connections among nodes, and
$\mathcal{A}=[a_{ij}]\in \mathbb{R}^{N\times N}$ is the weighted adjacency matrix with non negative elements $a_{ij}$. Each pair in $\mathcal{E}$ is called an \emph{edge}, where
$a_{ij}>0$ denotes an edge $(j,i)\in \mathcal{E}$ from node $j$ to
node $i$ with weight $a_{ij}$. Moreover, $a_{ij}=0$ if there is no
edge from node $j$ to node $i$. We assume there are no self-loops,
i.e.\ we have $a_{ii}=0$. A \emph{path} from node $i_1$ to $i_k$ is a
sequence of nodes $\{i_1,\ldots, i_k\}$ such that
$(i_j, i_{j+1})\in \mathcal{E}$ for $j=1,\ldots, k-1$. A \emph{directed tree} is a subgraph (subset
of nodes and edges) in which every node has exactly one parent node except for one node, called the \emph{root}, which has no parent node. The \emph{root set} is the set of root nodes. A \emph{directed spanning tree} is a subgraph which is
a directed tree containing all the nodes of the original graph. If a directed spanning tree exists, the root has a directed path to every other node in the tree.  

For a weighted graph $\mathcal{G}$, the matrix
$L=[\ell_{ij}]$ with
\[
\ell_{ij}=
\begin{system}{cl}
\sum_{k=1}^{N} a_{ik}, & i=j,\\
-a_{ij}, & i\neq j,
\end{system}
\]
is called the \emph{Laplacian matrix} associated with the graph
$\mathcal{G}$. The Laplacian matrix $L$ has all its eigenvalues in the
closed right half plane and at least one eigenvalue at zero associated
with right eigenvector $\textbf{1}$ \cite{royle-godsil}. Moreover, if the graph contains a directed spanning tree, the Laplacian matrix $L$ has a single eigenvalue at the origin and all other eigenvalues are located in the open right-half complex plane \cite{ren-book}.

\section{Problem Formulation}

Consider a MAS consisting of $N$ identical discrete-time linear dynamic agents with input delay:
\begin{equation}\label{eq1}
\begin{cases}
{x}_i(k+1)=Ax_i(k)+Bu_i(k-\kappa_i),\\
y_i(k)=Cx_i(k),\\
x_i(\psi)=\phi_i(\psi+\bar{\kappa}), \quad \psi\in\overline{[-\bar{\kappa},0]}
\end{cases}
\end{equation}
where $x_i(k)\in\mathbb{R}^n$, $y_i(k)\in\mathbb{R}^q$ and
$u_i(k)\in\mathbb{R}^m$ are the state, output, and the input of agent $i=1,\ldots, N$, respectively.

Moreover, $\kappa_i$ represent the input delays with $\kappa_i\in\overline{[0,\bar{\kappa}]}$, where $\bar{\kappa} = \max_{i} \{\kappa_i\}$, $\phi_i \in \ell_\infty^n(\bar{\kappa} )$ and the notation $\overline{[k_1, k_2]}$ means
\[\overline{[k_1, k_2]} = \{ k \in \mathbb{Z} : k_1 \leq k \leq k_2\}.\]

\begin{assumption}\label{Aass}
	We assume that:
	\begin{itemize}		
		\item[(i)] $(A,B)$ are stabilizable and $(C,A)$ are detectable.
		\item[(ii)] All eigenvalues of $A$ are in the closed unit disc. 
	\end{itemize}
\end{assumption}

In this paper, we consider regulated
state synchronization. The reference trajectory is generated by the following exosystem
\begin{equation}\label{solu-cond}
\begin{system*}{rl}
x_r(k+1) & = A x_r(k)\\
y_r(k)&=Cx_r(k).
\end{system*}	
\end{equation}
with $x_r(k)\in\R^n$.  Our objective is that the agents achieve regulated
state synchronization, that is
\begin{equation}\label{synchro}
\lim_{k\to \infty} (x_i(k)-x_r(k))=0,
\end{equation}
for all $i\in\{1,\ldots,N\}$. Clearly, we need some level of
communication between the exosystem and the agents.  We assume that a nonempty
subset $\mathscr{C}$ of the agents have access to their
own output relative to the output of the exosystem.  Specially, each
agent $i$ has access to the quantity
\begin{equation}\label{elp}
\psi_i=\iota_{i}(y_i(k)-y_r(k)),\qquad
\iota_i=
\begin{cases}
1, & i\in \mathscr{C},\\
0, & i\notin \mathscr{C}.
\end{cases}
\end{equation}
The network provides agent $i$ with the following information,
\begin{equation}\label{zeta_l-n}
\bar{\zeta}_i(k) = \sum_{j=1}^{N}a_{ij}(y_i-y_j)+\iota_{i}(y_i(k)-y_r(k)).
\end{equation}
where $a_{ij}\geq 0$ and $a_{ii}=0$. This communication topology of the network can be described by a weighted graph $\mathcal{G}$ with the $a_{ij}$ being the coefficients of the weighting matrix
$\mathcal{A}$ (not of the dynamics matrix $A$ introduced in\eqref{eq1}). 

We refer to \eqref{zeta_l-n} as \emph{partial-state coupling} since only part of
the states are communicated over the network. When $C=I$, it means all states are communicated over the network, we call it \emph{full-state coupling}. Then, the original agents are expressed as
\begin{equation}\label{neq1}
x_i(k+1)=Ax_i(k)+Bu_i(k-\kappa_i)
\end{equation}
meanwhile, \eqref{zeta_l-n} will change as
\begin{equation}\label{zeta_l-nn}
\bar{\zeta}_i(k) = \sum_{j=1}^{N}a_{ij}(x_i(k)-x_j(k))+\iota_{i}(x_i(k)-x_r(k)).
\end{equation}

To guarantee that each agent can achieve the required regulation, we need to make sure that there exists a path to each node starting with node from the set $\mathscr{C}$. Motivated by this requirement, we define the following set of graphs.
\begin{definition}\label{graph-def}
	Given a node set $\mathscr{C}$, we denote by $\mathbb{G_\mathscr{C}^N}$ the set of all graphs with $N$ nodes containing the node set $\mathscr{C}$, such that every node of the network graph $\mathscr{G}\in \mathbb{G_\mathscr{C}^N}$ is a member of a directed tree which has its root contained in the node set $\mathscr{C}$.
\end{definition}
\begin{remark}
	Note that Definition \ref{graph-def} does not require necessarily the existence of directed spanning tree.
\end{remark}
From now on, we will refer to the node set $\mathscr{C}$ as the \emph{root set} in view of Definition \ref{graph-def}. For any graph $\mathcal{G}\in \mathbb{G}^N_{\mathscr{C}}$, with the
associated Laplacian matrix $L$, we define the expanded Laplacian
matrix as
\[
\bar{L}=L+\diag\{\iota_i\} = [\bar{\ell}_{ij}]_{N\times N}.
\]
and we define
\begin{equation}\label{hodt-LDa}
\bar{D}=I-(2I+D_{\text{in}})^{-1}\bar{L}.
\end{equation}
where
\[
D_{\text{in}}=\diag\{d_{in}(i) \}
\]
with $d_{in}(i)=\sum_{j=1}^N a_{ij}$.  It is easily verified that the
matrix $\bar{D}$ is a matrix with all elements nonnegative and the sum of each row is less than or equal to $1$. Note that based on \cite[Lemma 1]{liu2018regulated}, matrix $\bar{D}$ has all eigenvalues in the open unit disc if and only if $\mathscr{G}\in \mathbb{G_\mathscr{C}^N}$. 

We can obtain the new information exchange
\begin{equation}\label{zeta_l-nd}
\bar{\zeta}_i^d(k) = \frac{1}{2+d_{in}(i)}\sum_{j=1}^{N}a_{ij}(y_i(k)-y_j(k))+\iota_{i}(y_i(k)-y_r(k)).
\end{equation}
and
\begin{equation}\label{zeta_l-nnd}
\bar{\zeta}_i^d(k) = \frac{1}{2+d_{in}(i)}\sum_{j=1}^{N}a_{ij}(x_i(k)-x_j(k))+\iota_{i}(x_i(k)-x_r(k))
\end{equation}

In this paper, we introduce an additional
information exchange among protocols. In particular, each agent 
$i=1,\ldots, N$ has access to additional information, denoted by
$\hat{\zeta}_i(k)$, of the form
\begin{equation}\label{eqa1}
\hat{\zeta}_i(k)=\frac{1}{2+d_{in}(i)}\sum_{j=1}^N\bar{\ell}_{ij}\xi_j(k)
\end{equation}
where $\xi_j(k)\in\mathbb{R}^n$ is a variable produced internally by agent $j$ and to be defined in next sections.

We formulate the problem for regulated state
synchronization of a MAS with full- and partial-state coupling as following.

\begin{problem}\label{prob3}
	Consider a MAS described by \eqref{eq1} satisfying Assumption \ref{Aass}, with a given $\bar{\kappa}$ and the
	associated exosystem \eqref{solu-cond}. 
	Let a set of nodes
	$\mathscr{C}$ be given which defines the set
	$\mathbb{G}_{\mathscr{C}}^N$ and 	
	let the asssociated network
	communication graph $\mathcal{G} \in \mathbb{G}_{\mathscr{C}}^N$ be
	given by \eqref{zeta_l-n}.
	
	The \textbf{scalable regulated state synchronization problem with additional information exchange} of a discrete-time MAS is to find,
	if possible, a linear dynamic protocol for each agent $i \in \{1,\hdots,N\}$, using only knowledge of agent model, i.e., $(A,B,C)$, and upper bound of delays $\bar{\kappa}$, of the form:
	\begin{equation}\label{protoco4}
	\begin{system}{cl}
	x_{c,i}(k+1)&=A_{c,i}x_{c,i}(k)+B_{c,i}u_i(k-\kappa_i)\\&\qquad+C_{c,i}\bar{\zeta}_i^d(k)+D_{c,i}\hat{\zeta}_i(k),\\
	u_i(k)&=F_{c,i}x_{c,i}(k),
	\end{system}
	\end{equation}
	where $\hat{\zeta}_i(k)$ is defined in \eqref{eqa1} with $\xi_i(k)=H_{c}x_{i,c}(k)$, and $x_{c,i}(k)\in\R^{n_i}$, such that regulated
	state synchronization \eqref{synchro} is achieved for any $N$ and any
	graph $\mathcal{G}\in \mathbb{G}_{\mathscr{C}}^N$.
\end{problem}

%

\section{Protocol Design}

In this section, we will consider the regulated state synchronization problem for a MAS with input delays. In particular, we cover separately systems with full-state coupling and those with partial-state coupling.

\subsection{Full-state coupling}

Firstly, we define
\[
\omega_{\max}=
\begin{cases}
0, \qquad\text{A is Schur stable},& \\
\max\{\omega\in[0,\pi] | \det(e^{j\omega} I-A)=0\}, &\text{otherwise}.
\end{cases}
\]

Then, we design a dynamic protocol with additional information exchanges for agent
$i\in\{1,\ldots,N\}$ as follows.
\begin{equation}\label{pscp1}
\begin{system}{cll}
\chi_i(k+1) &=& A\chi_i(k)+Bu_i(k-\kappa_i)+A\bar{\zeta}_i^d(k)-A\hat{\zeta}_i(k) \\
u_i(k) &=& - \rho K_{\eps} \chi_i(k),
\end{system}
\end{equation}
where $\rho>0$ and
\[
K_{\eps}=(I+B\T P_\eps B)^{-1}B\T P_\eps A
\]
and $\eps$ is a parameter satisfying
$\eps\in (0,1]$, $P_{\eps}$ satisfies
\begin{equation}\label{arespecial}
A\T P_{\eps}A - P_{\eps} -  A\T P_{\eps}B(I +B\T P_{\eps} B)^{-1}B\T
P_{\eps} A +  \eps I = 0 
\end{equation}
Note that for any $\eps>0$, there exists a unique solution of \eqref{arespecial}.

The agents communicate $\xi_i(k)$, which are chosen as $\xi_i(k)=\chi_i(k)$, therefore each agent has access to the following information:
\begin{equation}\label{info1}
\hat{\zeta}_i(k)=\frac{1}{2+d_{in}(i)}\sum_{j=1}^N\bar{\ell}_{ij}\chi_j(k).
\end{equation}
while $\bar{\zeta}_i^d(k)$ is defined by \eqref{zeta_l-nnd}.
\begin{remark}
	\eqref{arespecial} is a special case of the general low-gain $H_2$ discrete algebraic Riccati equation ($H_2$-DARE), which is written as follows:
	\begin{equation}\label{aregeneral}
	A\T P_{\eps}A - P_{\eps} -  A\T P_{\eps}B(R_{\eps}+B\T P_{\eps} B)^{-1}B\T
	P_{\eps} A +  Q_{\eps} = 0 
	\end{equation}
	where $R_{\eps} > 0$, and $Q_{\eps} > 0$ is such that $Q_{\eps} \rightarrow 0$ as $\eps \rightarrow 0.$ In our case, we restrict our attention to $Q_{\eps} = \eps I$ and $R_{\eps} = I.$ However, as shown in \cite{saberi-stoorvogel-sannuti-exter}, when $A$ is neutrally stable, there exists a suitable (nontrivial) choice of $Q_{\eps}$ and $R_{\eps}$ which yields an explicit solution of \eqref{aregeneral}, of form
	\begin{equation}\label{neutral}
	P_{\eps} = \eps P
	\end{equation}
	where $P$ is a positive definite matrix that satisfies $A^T P A \leq P.$
\end{remark}
Our formal result is stated in the following theorem.
\begin{theorem}\label{mainthm1}
	Consider a MAS described by \eqref{neq1} satisfying Assumption \ref{Aass}, with a given $\bar{\kappa}$ and the
	associated exosystem \eqref{solu-cond}. 
	Let a set of nodes
	$\mathscr{C}$ be given which defines the set
	$\mathbb{G}_{\mathscr{C}}^N$ and 	
	let the asssociated network
	communication graph $\mathcal{G} \in \mathbb{G}_{\mathscr{C}}^N$ be
	given by \eqref{zeta_l-nnd}.
	
	Then the scalable regulated state synchronization problem as stated in Problem
	\ref{prob3} is solvable if
	\begin{equation}\label{boundtau}
	\bar{\kappa}\omega_{\max}<\frac{\pi}{2}. 
	\end{equation}
	In particular, there exist $\rho>0.5$ and $\eps^* > 0$ that depend only on $\bar{\kappa}$ and the agent models such that, for any $\eps\in(0,\eps^*]$, the dynamic protocol given by \eqref{pscp1} and \eqref{arespecial} solves the scalable regulated state
	synchronization problem for any $N$ and any graph
	$\mathcal{G}\in\mathbb{G}_{\mathscr{C}}^N$.
\end{theorem}

To obtain this result, we need the following lemma.

\begin{lemma}[\cite{zhang-saberi-stoorvogel-continues-discrete}]\label{hode-lemma-ddelay}
	Consider a linear time-delay system
	\begin{equation}\label{hoded-system1}
	x(k+1)=Ax(k)+\sum_{i=1}^{m}A_{i}x(k-\kappa_{i}),
	\end{equation}
	where $x(k)\in\R^{n}$ and $\kappa_{i}\in\N^+$. Suppose
	$A+\sum_{i=1}^{m}A_{i}$ is Schur stable. Then, \eqref{hoded-system1}
	is asymptotically stable if
	\[
	\det[e^{j\omega}I-A-
	\sum_{i=1}^{m}e^{-j\omega\kappa_i^r}A_{i}]\neq 0,
	\]
	for all $\omega\in[-\pi,\pi]$ and for all $\kappa_i\in\overline{[0,\bar{\kappa}]}$ \text{ for} ($i=1,\ldots, N$).
\end{lemma}

\begin{lemma}[\cite{lee-kim-shim}]\label{lemma-full}
	Consider a linear uncertain system,
	\begin{equation}\label{linunc}
	x(k+1)=Ax(k)+\lambda Bu(k), \qquad x(0)=x_{0},
	\end{equation}
	where $\lambda\in \mathbb{C}$ is unknown. Assume that $(A,B)$ is
	stabilizable and $A$ has all its eigenvalues in the closed unit
	disc. A low-gain state feedback $u=F_{\delta}x$ is constructed,
	where 
	\begin{equation}\label{fdelta}
	F_{\delta}=-(B\T P_{\delta}B+I)^{-1}B\T P_{\delta}A, 
	\end{equation}
	with $P_{\delta}$ being the unique positive definite solution of the
	$H_{2}$ algebraic Riccati equation,
	\begin{equation}\label{hoded-ARE-sept-17}
    P_\delta=A\T P_\delta A + \delta I - A\T P_\delta B(B\T P_\delta B+I)^{-1}B\T P_\delta A. 
    \end{equation}
	Then, $A+\lambda BF_{\delta}$ is Schur stable for any $\lambda\in \C$
	satisfying,
	\begin{equation}\label{hoded-stablefield}
	\lambda\in \Omega_{\delta}:=\left \{ z\in \mathbb{C}: \left |
	z-\left (1+\tfrac{1}{\gamma_{\delta}}\right ) \right
	|<\tfrac{\sqrt{1+\gamma_{\delta}}}{\gamma_{\delta}}\right \}, 
	\end{equation}
	where $\gamma_{\delta}=\lambda_{\max}(B\T P_{\delta}B)$.  As
	$\delta\rightarrow 0$, $\Omega_{\delta}$ approaches the set
	\[
	H_1:=\{z\in \C :\re z >\tfrac{1}{2}\}
	\]
	in the sense that any compact subset of $H_1$ is contained in
	$\Omega_{\delta}$ for a $\delta$ small enough.
\end{lemma}

\begin{proof}[Proof of Theorem \ref{mainthm1}]
	Firstly, let $\tilde{x}_i=x_i-x_r$, we have
	\[
	{\tilde{x}}_i(k+1)=A{\tilde{x}}_i(k)+Bu_i(k-\kappa_i)
	\]

	We define 
	\begin{align*}
	&\tilde{x}(k)=\begin{pmatrix}
	\tilde{x}_1(k)\\\vdots\\\tilde{x}_N(k)
	\end{pmatrix} , 
	\chi(k)=\begin{pmatrix}
	\chi_1(k)\\\vdots\\\chi_N(k)
	\end{pmatrix},  \\
	&\tilde{x}^\kappa(k)=\begin{pmatrix}
	\tilde{x}_1(k-\kappa_1)\\\vdots\\\tilde{x}_N(k-\kappa_N)
	\end{pmatrix},\text{ and }  
	\chi^\kappa(k)=\begin{pmatrix}
	\chi_1(k-\kappa_1)\\\vdots\\\chi_N(k-\kappa_N)
	\end{pmatrix}
	\end{align*}
	then we have the following closed-loop system
	\begin{equation}
	\begin{system*}{ll}
	{\tilde{x}}(k+1)=&(I\otimes A) \tilde{x}(k)- \rho(I\otimes BK_\eps)\chi^\kappa(k)\\
	{\chi}(k+1)=&(I\otimes A) \chi(k)-\rho(I\otimes BK_\eps)\chi^\kappa(k)\\
	&\qquad+[(I-\bar{D})\otimes A](\tilde{x}(k)-\chi(k)).
	\end{system*}
	\end{equation}

	Let $e(k)=\tilde{x}(k)-\chi(k)$, we can obtain  
	\begin{equation}\label{newsystem2}
	\begin{system*}{ll}
	{\tilde{x}}(k+1)=&(I\otimes A) \tilde{x}(k)- \rho (I\otimes BK_\eps)\tilde{x}^\kappa(k)\\
	&\qquad+ \rho (I\otimes BK_\eps)e^\kappa(k)\\
	{e}(k+1)=&(\bar{D}\otimes A)e(k)
	\end{system*}
	\end{equation}
	where $e^\kappa(k)=\tilde{x}^\kappa(k)-\chi^\kappa(k)$.
	The proof proceeds in two steps.
	
	\textbf{Step 1:}  First, we prove the stability of system \eqref{newsystem2} without delays, i.e.	
	\begin{equation}\label{newsystem22}
	\begin{system*}{l}
	{\tilde{x}}(k+1)=(I\otimes A) \tilde{x}(k)- \rho (I\otimes BK_\eps)\tilde{x}(k)+ \rho (I\otimes BK_\eps)e(k)\\
	{e}(k+1)=(\bar{D}\otimes A)e(k)
	\end{system*}
	\end{equation}
	where $\bar{D}=[\bar{d}_{ij}]\in\R^{N\times N}$ and we have that the eigenvalues of $\bar{D}$  are in open unit disk.
	The eigenvalues of $\bar{D}\otimes A$ are of the form
	$\lambda_i \mu_j$, with $\lambda_i$ and $\mu_j$ eigenvalues of
	$\bar{D}$ and $A$, respectively \cite[Theorem 4.2.12]{horn-johnson}. Since $|\lambda_i|<1$ and
	$|\mu_j|\leq 1$, we find $\bar{D}\otimes A$ is Schur stable. Then we have
	\begin{equation}\label{estable}
	\lim_{k\to \infty}e_i(k)\to 0
	\end{equation}
	Therefore, we have that the dynamics for $e_i(k)$ is asymptotically stable.

	According to the above result, for \eqref{newsystem22} we just need to prove the stability of 
	\[
	{\tilde{x}}(k+1)=[I\otimes (A- \rho BK_\eps)]\tilde{x}(k)
	\]
	or the stability of 
	\[
	A- \rho BK_\eps
	\]
	
	Based on Lemma \ref{lemma-full}, there exist $\rho >0.5$ and $\eps* >0$ such that $A- \rho BK_\eps$ is Schur stable
	for $\eps \in (0,\eps^*]$.
	
	\textbf{Step 2:} In this step, we consider \eqref{newsystem2}, i.e. system in
	the presence of delays. Then, we rewrite \eqref{newsystem2} as
	\[
	\begin{pmatrix}
	{\tilde{x}}(t+1)\\
	{\bar{e}}(t+1)
	\end{pmatrix}=A_0^d
	\begin{pmatrix}
	{\tilde{x}}(t)\\
	{\bar{e}}(t)
	\end{pmatrix}+A_1^d\begin{pmatrix}
	{\tilde{x}}^\tau(t)\\
	{\bar{e}}^\tau(t)
	\end{pmatrix}
	\]
	where
	\begin{align*}
	&A_0^d=\begin{pmatrix}
	I\otimes A&0\\
	0&\bar{D}\otimes A
	\end{pmatrix}\\
	&A_1^d=\begin{pmatrix}
	-\rho_d I\otimes
	BK_\eps&\rho_d I\otimes BK_\eps\\
	0&0
	\end{pmatrix}.
	\end{align*}
	Thus, from Lemma \ref{hode-lemma-ddelay}, we need to prove 
	\begin{equation}\label{bound-conditiond1}
	\det[e^{j\omega} I-A_0^d- A_{e^\tau} A_1^d]\neq 0. 
	\end{equation}
	Since $A_0^d$ and $A_1^d$ are upper triangular matrices, we can
	rewrite \eqref{bound-conditiond1} as
	\[
	\begin{system*}{l}
	\det[e^{j\omega} I-I\otimes A+ \rho_d e^{\tau}\otimes BK_\eps]\neq 0,\\
	\det[e^{j\omega} I-\bar{D}\otimes A]\neq 0.
	\end{system*} 
	\]
	Due to the fact that $\bar{D}\otimes A$ is Schur stable,
	we just need to prove
	\[
	\det[e^{j\omega} I-I\otimes A+ \rho_d e^{\tau}\otimes BK_\eps]\neq 0 
	\]
	or 
	\begin{equation}\label{bound-condition}
	\det[e^{j\omega} I-A+ \rho e^{-j\omega\kappa_i}BK_\eps]\neq 0
	\end{equation}
	for $\omega\in [-\pi,\pi]$ and $\kappa_i\in\overline{[0,\bar{\kappa}]}$. We choose $\rho$, such that  
	\begin{equation}	
	\rho\cos(\bar{\kappa}\omega_{\max}) >\frac{1}{2}.
	\end{equation}
	Let $\rho$ be fixed. Meanwhile, we note that there exists a $\theta$ such that
	\[
	\rho>\frac{1}{2\cos(\bar{\kappa}\omega)}, \forall |\omega|<\omega_{\max}+\theta
	\]
	
	Then, we split the proof of \eqref{bound-condition} into
	two cases where $\pi\geq|\omega|\geq \omega_{\max}+\theta$ and $|\omega|<\omega_{\max}+\theta$ respectively.
	
	If $\pi\geq|\omega|\geq \omega_{\max}+\theta$, we have
	$\det(e^{-j\omega}I-A)\neq 0$, which yields $\sigma_{\min}(e^{j\omega} I-A)>0$. Because $\sigma_{\min}(e^{j\omega} I-A)$ depends continuously on $\omega$ and the set $\{\pi\geq|\omega|\geq \omega_{\max}+\theta\}$ is compact. Hence, there exists a $\mu>0$ such
	that
	\[
	\sigma_{\min}(e^{j\omega} I-A)>\mu,\quad \forall \omega \text{ such
		that } |\omega|\geq \omega_{\max}+\theta.
	\]
	Given $\rho$, there exists $\eps^*>0$ such that $\|\rho e^{-j\omega\kappa_i} BK_{\eps}\|\leq \mu/2$. Then, we obtain
	\[
	\sigma_{\min}(e^{j\omega}I-A-\rho e^{-j\omega\kappa_i}B K_\eps)\geq \mu-\tfrac{\mu}{2}\geq \tfrac{\mu}{2}.
	\]
	Therefore, condition
	\eqref{bound-condition} holds for
	$\pi\geq|\omega|\geq \omega_{\max}+\theta$.
	
	Now, it remains to show that condition \eqref{bound-condition} holds for $|\omega|<\omega_{\max}+\theta$. We find that
	\[
	-\omega\kappa_i<|\omega| \bar{\kappa}\leq\dfrac{\pi}{2},
	\]
	and hence 
	\[
	\rho \cos(-\omega\kappa_i)>\rho \cos(|\omega| \bar{\kappa})>\frac{1}{2}.
	\]
	It implies that there exists a small enough $\eps$ such that
	\[
	A- \rho e^{-j\omega\kappa_i}BK_\eps
	\]
	is Schur stable based on Lemma \ref{lemma-full} for the fixed $\rho$.
	Therefore,  \eqref{bound-condition} 
	holds for $|\omega|<\omega_{\max}+\theta$ for a small enough $\eps$ and a fixed $\rho$ satisfying
	\[
	\rho>\frac{1}{2\cos(\bar{\kappa}\omega_{\max} )}.
	\] 
	Thus, we can obtain the regulated state synchronization result based on Lemma \ref{hode-lemma-ddelay}.		
\end{proof}
\subsection{Partial-state coupling}
In this subsection, we will consider the case via partial-state coupling.
We design the following dynamic protocol with additional information exchanges as follows.
\begin{equation}\label{pscp3}
\begin{system}{cll}
\hat{x}_i(k+1) &=& A\hat{x}_i(k)+B\hat{\zeta}_{i2}(k)+F(\bar{\zeta}_i^d(k)-C\hat{x}_i(k)) \\
\chi_i(k+1) &=& A\chi_i(k)+Bu_i(k-\kappa_i)+A\hat{x}_i(k)-A\hat{\zeta}_{i1}(k)\\
u_i(k) &=& -\rho K_{\eps}\chi_i(k),
\end{system}
\end{equation}
for $i=1,\ldots,N$ where $F$ is a pre-design matrix such that $A-FC$ is Schur stable,
\[
K_\eps=(I+B\T P_\eps B)^{-1}B\T P_\eps A,
\]
and $\rho>0$. $\eps$ is a parameter satisfying
$\eps\in (0,1]$, $P_{\eps}$ satisfies \eqref{arespecial} and is the unique solution of \eqref{arespecial} for any $\eps>0$. $\rho$ and $\omega_{\max}$ are defined in \eqref{pscp1}.
In this protocol, the agents communicate $\xi_i=(\xi_{i1}\T,\xi_{i2}\T)\T$ where $\xi_{i1}(k)=\chi_i(k)$ and $\xi_{i2}(k)=u_i(k-\kappa_i)$, therefore each agent has access to the additional information $\hat{\zeta}_i=(\hat{\zeta}_{i1}\T,\hat{\zeta}_{i2}\T)\T$:

\begin{equation}\label{add_1}
\hat{\zeta}_{i1}(k)=\frac{1}{2+d_{in}(i)}\sum_{j=1}^N\bar{\ell}_{ij}\chi_j(k),
\end{equation}
and
\begin{equation}\label{add_2}
\hat{\zeta}_{i2}(k)=\frac{1}{2+d_{in}(i)}\sum_{j=1}^{N}\bar{\ell}_{ij}u_j(k-\kappa_j).
\end{equation}
$\bar{\zeta}_i^d(k)$ is also defined as \eqref{zeta_l-nd}. Then we have the following theorem for MAS via partial-state coupling.

\begin{theorem}\label{mainthm2}
	Consider a MAS described by \eqref{eq1} satisfying Assumption \ref{Aass}, with a given $\bar{\kappa}$ and the
	associated exosystem \eqref{solu-cond}. 
	Let a set of nodes
	$\mathscr{C}$ be given which defines the set
	$\mathbb{G}_{\mathscr{C}}^N$ and 	
	let the asssociated network
	communication graph $\mathcal{G} \in \mathbb{G}_{\mathscr{C}}^N$ be
	given by \eqref{zeta_l-nd}.
	
	Then the scalable regulated state synchronization problem as stated in Problem
	\ref{prob3} is solvable if \eqref{boundtau} holds.
	In particular, there exist $\rho>1$ and $\eps^* > 0$ that depend only on $\bar{\kappa}$ and the agent models such that, for any $\eps\in(0,\eps^*]$, the dynamic protocol given by \eqref{pscp3} and \eqref{arespecial} solves the scalable regulated state
	synchronization problem for any $N$ and any graph
	$\mathcal{G}\in\mathbb{G}_{\mathscr{C}}^N$.
\end{theorem}
%

\begin{proof}[Proof of Theorem \ref{mainthm2}]
	Similar to Theorem \ref{mainthm1}, let $\tilde{x}_i(k)=x_i(k)-x_r(k)$, we have
	\[
	\begin{system}{cll}
	\tilde{x}_i(k+1)&=&A{\tilde{x}}_i(k)+Bu_i(k-\kappa_i)\\
	\hat{x}_i(k+1) &=& A\hat{x}_i(k)+B\hat{\zeta}_{i2}(k)+F(\bar{\zeta}_i^d(k)-C\hat{x}_i(k)) \\
	{\chi}_i(k+1) &=& A\chi_i(k)+Bu_i(k-\kappa_i)+\hat{x}_i(k)-\hat{\zeta}_{i1}(k)
	\end{system}
	\]
	
	Then we have the following closed-loop system
	\begin{equation}
	\begin{system*}{ll}
	\tilde{x}(k+1)=&(I\otimes A) \tilde{x}(k)-\rho (I\otimes BK_\eps)\chi^\kappa(k)\\
	\hat{x}(k+1) =& I\otimes (A-FC)\hat{x}(k)-\rho[(I-\bar{D})\otimes B K_\eps]\chi^\kappa(k)\\
	&\qquad\qquad+[(I-\bar{D})\otimes FC]\tilde{x}(k) \\
	\chi(k+1) =& [(I-\bar{D})\otimes A]\chi(k)-\rho(I\otimes BK_\eps)\chi^\kappa(k)+\hat{x}(k)
	\end{system*}
	\end{equation}
	by defining $e=\tilde{x}-\chi$ and $\bar{e}=[(I-\bar{D})\otimes I]\tilde{x}-\hat{x}$, we obtain  
	\begin{equation}\label{newsystem3}
	\begin{system*}{l}
	\tilde{x}(k+1)=(I\otimes A) \tilde{x}(k)-\rho(I\otimes BK_\eps)\tilde{x}^\kappa(k)+\rho(I\otimes BK_\eps)e^\kappa(k)\\
	\bar{e}(k+1)=I\otimes (A-FC)\bar{e}(k)\\
	e(k+1)=(\bar{D}\otimes A)e(k)+\bar{e}(k)
	\end{system*}
	\end{equation}
	
	Similar to Theorem \ref{mainthm1}, we prove the stability of \eqref{newsystem3} without delays first,
	\begin{equation}\label{newsystem33}
	\begin{system*}{l}
	\tilde{x}(k+1)=(I\otimes A) \tilde{x}(k)-\rho(I\otimes BK_\eps)\tilde{x}+ \rho(I\otimes BK_\eps)e(k)\\
	\bar{e}(k+1)=I\otimes (A-FC)\bar{e}(k)\\
	e(k+1)=(\bar{D}\otimes I)e(k)+\bar{e}(k)
	\end{system*}
	\end{equation}
	
	Since we have $A-FC$ and $\bar{D}\otimes A$ are Schur stable, one can obtain
	\[
	\lim_{k\to \infty}\bar{e}(k)\to 0 \text{ and }\lim_{k\to \infty}e(k)\to 0
	\]
	i.e. we just need to prove the stability of 
	\[		
	\tilde{x}_i(k+1)=(A-\rho BK_\eps)\tilde{x}_i(k).
	\]		
	Similarly to Theorem \ref{mainthm1}, we can obtain the result about the stability of the above system directly by using Lemma \ref{hode-lemma-ddelay} and choosing
	\[
	\rho> \frac{1}{2}.
	\]
	
	Next, similar to the proof of Theorem \ref{mainthm1}, we just need to prove \eqref{bound-condition} for $\omega\in [-\pi, \pi]$ and $\kappa_i\in\overline{[0,\bar{\kappa}]}$ based on Lemma \ref{hode-lemma-ddelay}. 
	
	Then, we can obtain the synchronization result for a small enough $\eps$ and a fixed $\rho$ satisfying
	\[
	\rho>\frac{1}{2\cos(\bar{\kappa}\omega_{\max} )}.
	\] 
\end{proof}
\begin{figure}[t]
	\includegraphics[width=9cm, height=7.2cm]{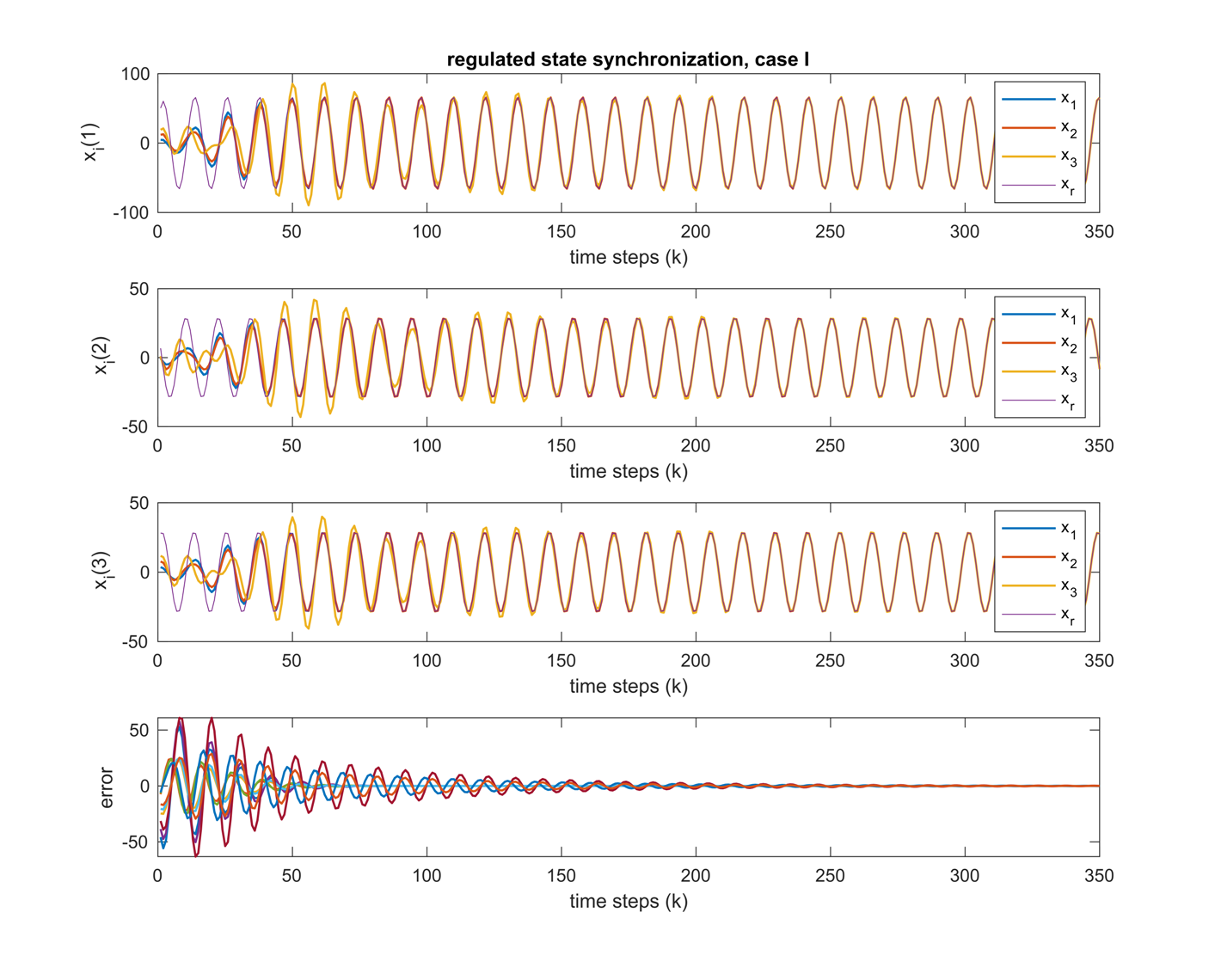}	
	\centering
	\vspace{-1cm}
	\caption{Directed communication network with $3$ nodes and full-state coupling.}\label{graph_3nodesfull}
\end{figure}\begin{figure}[t!]
	\includegraphics[width=9cm, height=7.2cm]{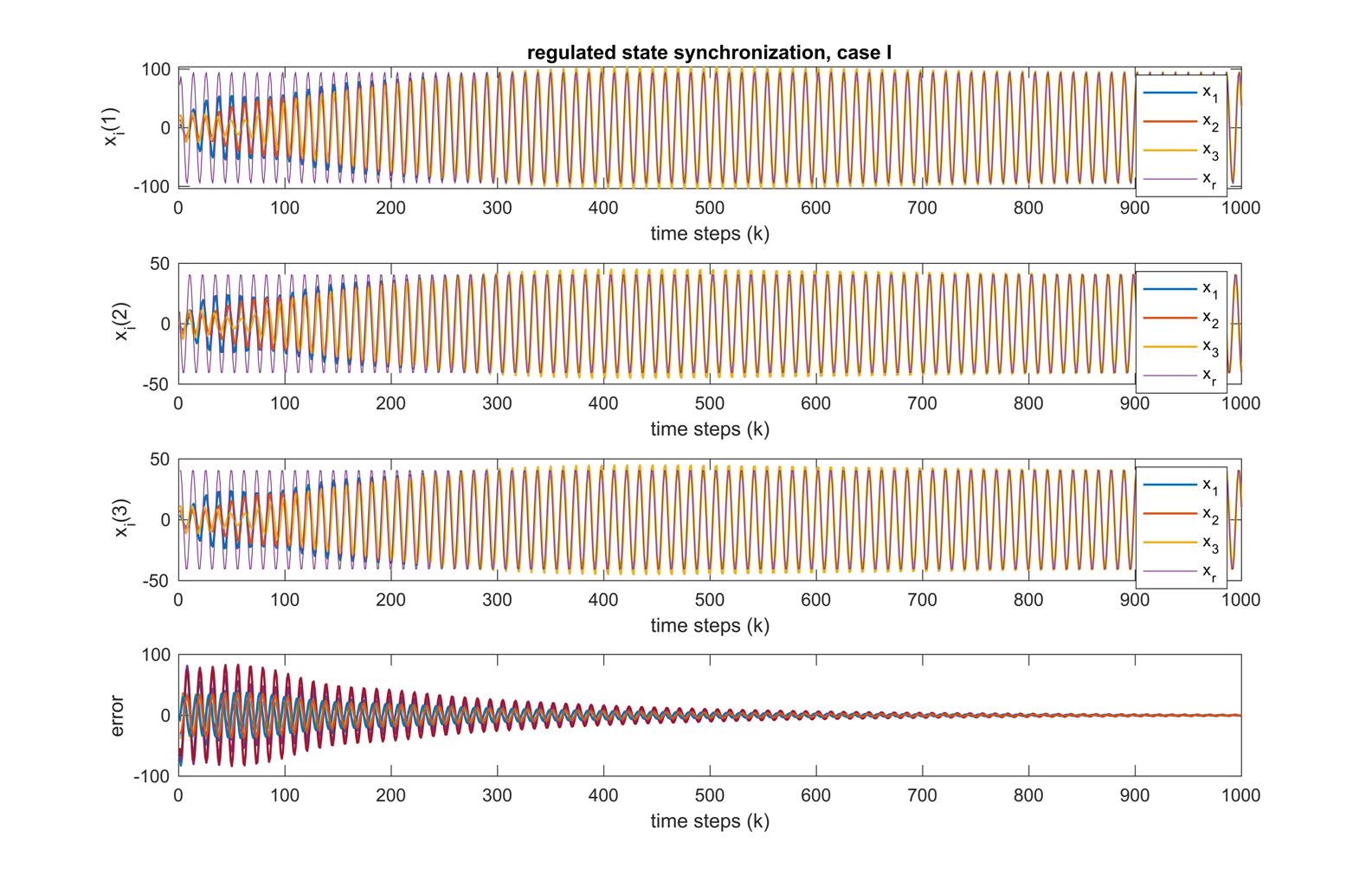}
	\centering
	\vspace{-1cm}
	\caption{Directed communication network with $3$ nodes and partial-state coupling.}\label{graph_3nodespartial}
\end{figure}\begin{figure}[t!]
	\includegraphics[width=9cm, height=7.2cm]{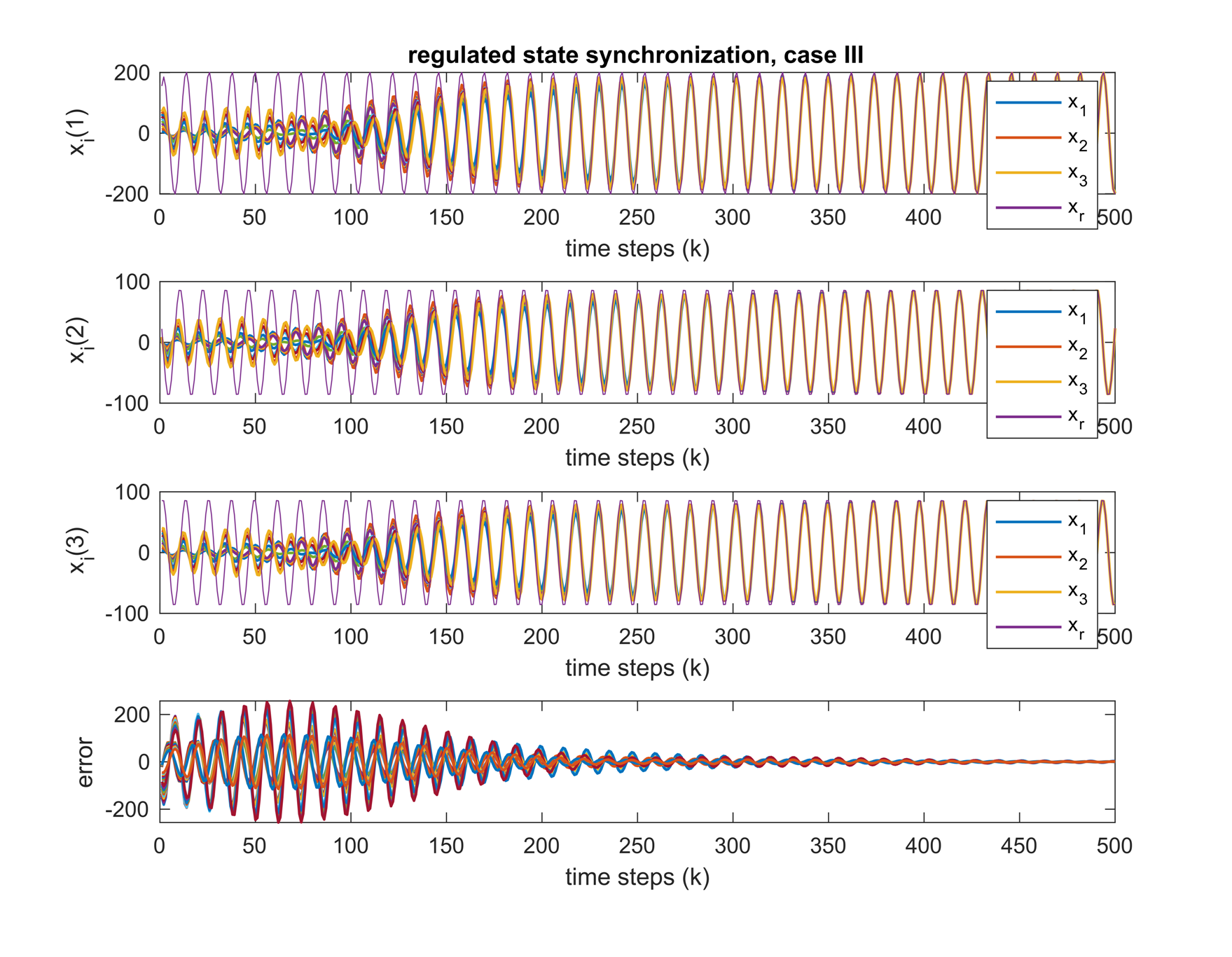}
	\centering
	\vspace{-1cm}
	\caption{Directed communication network with $10$ nodes and full-state coupling.}\label{10nodes_full}
\end{figure}\begin{figure}[t!]
	\includegraphics[width=9cm, height=7.2cm]{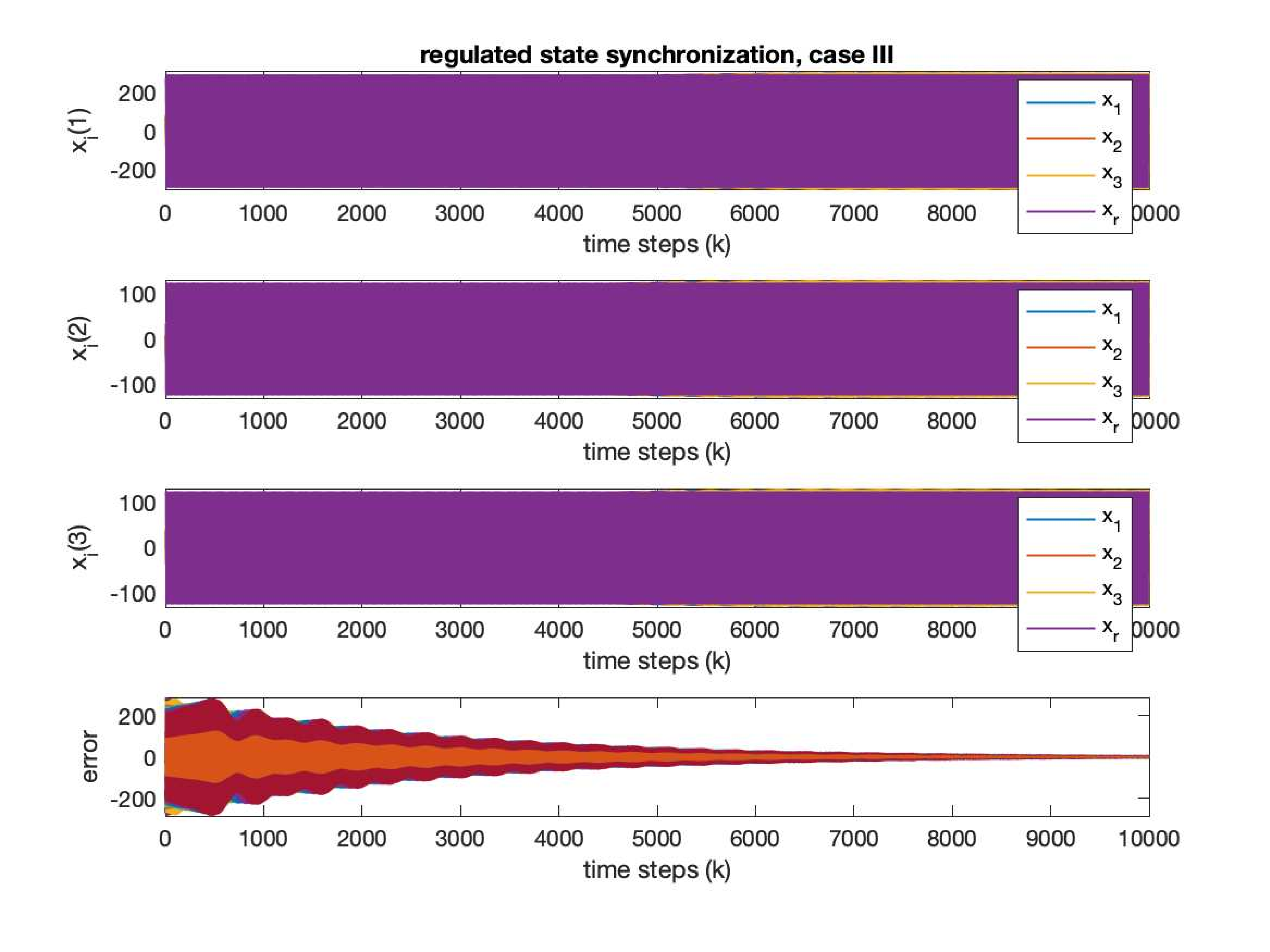}
	\centering
	\vspace{-1cm}
	\caption{Directed communication network with $10$ nodes and partial-state coupling.}\label{10nodes_partial}
\end{figure}
\section{MATLAB Implementation and Example}
\subsection{Implementation}
We present and discuss our MATLAB source code that implements the general methods described in this paper. Namely, we include three files that allow the user to make use of our protocol in a number of different contexts. The code (and hence, this section) is very similar to our continuous-time version of this paper \cite{inputDelay-Con-ArXive}. We will also highlight a few key differences.

The file \verb+discrete_protocol_design.m+ designs the central product of this paper, the protocol, setting it up for use. It accepts \emph{only} the agent model ($A,B,C$) and an upper bound on the delays ($\bar{\kappa}$). Recall that in the full-state coupling case, the protocol is given by \eqref{pscp1} and \eqref{arespecial} and in the partial-state coupling case by \eqref{pscp3} and \eqref{arespecial}.
In view of these equations, the function \verb+discrete_protocol_design+ returns the relevant data necessary to define them. Namely, $\epsilon^*$, $\rho,$ $\bar{\kappa}_{\max}$, $K$, and $F$. Note that the selection of $F$ is arbitrary, and we welcome the user to change our code to pick any value as far as $A - FC$ is Schur stable. 

We describe briefly how this function operates. The proof of Theorem 1, particularly the definitions of $\epsilon^*$ and $\rho$, reveal that there is a large degree of freedom in choosing the pair $(\epsilon^*, \rho).$ Furthermore, different choices can lead to drastically different speeds of convergence (i.e., how quickly $x_i$ converges to $x_r$ for $i = 1, \dots, N$). For fixed $\rho$, among valid choices of $\epsilon^*$, faster convergence is typically obtained by larger $\epsilon^*.$ In this vein, our algorithm seeks to obtain a less conservative $\epsilon^*$, given a fixed $\rho$. This is done by first choosing $\theta$ based on $\rho$, from which we make a non-conservative estimate of $\mu$, finally choosing $\epsilon^*$ based on its definition (which involves $\mu$). Moreover, we comment that this file in no way chooses the best $(\epsilon^*, \rho)$ pair that optimizes convergence. We simply guarantee that our parameters satisfy the solvability conditions laid out in this paper. The existence of an algorithm that chooses optimal $(\epsilon^*, \rho)$ in all generality (both in discrete-time and continuous-time) remains an open question, and will be the subject of future research.

The second and final main file we include is the \verb+ discrete_input_delay_solver.m+  file, which is a complete simulation package. This file defines a function of the same name that accepts an agent model, the delays, the adjacency matrix of the communication network, the set of leader nodes, and the initial conditions (in addition to the period of integration $K_{\max}$, which specifies the time interval $\overline{[0,K_{\max}]}$ that the user wants the solution over). From here, the function uses \verb+protocol_design+ to choose acceptable protocol parameters to achieve regulated state synchronization as stated in \eqref{synchro} through the use of protocols \eqref{pscp1} and \eqref{arespecial} in the case of full-state coupling and \eqref{pscp3} and \eqref{arespecial} for partial-state coupling. If matrix $C$ passed to the function is the identity, the protocol for full-state coupling will be enacted, otherwise, partial state coupling protocol will be utilized.  The underlying algorithm is one of the few things about our source code that is truly different from the continuous-time version, and is much simpler. As opposed to having to choose a mesh (i.e., a very dense time series) to get the solution over, we need only get the solution at time steps $k = 1, \dots, K_{\max}$. Moreover, all that is required to accomplish this is a simple update loop, where for each $k = 0, \dots, K_{\max} - 1,$ $x(k+1)$ is computed from $x(k)$ based on the agent dynamics \eqref{eq1}, as well as the appropriate protocol. In our implementation, we write the two together as a closed-loop system. In the end, function returns the state $x$ and the exosystem $x_r$ as matrices, with column $k$ representing $x(k)$ and $x_r(k),$ respectively. This gives data which can be easily plotted or used for a variety of purposes.

Finally, we include a third file \verb+plotting.m+, which allows the user to visualize the simulation results (without worrying about formatting, or other logistical issues). We hope the inclusion of these files will allow the reader to illustrate our results for themselves if so desired, and more importantly, that those who have use for a product of this form will profit from them in their own endeavors.
\subsection{Numerical example}
Consider an agent model \eqref{eq1} with
\[ A=
\begin{pmatrix}
1/2 & 1 & 1 \\
0 & \sqrt{3}/2 & -1/2 \\
0 & 1/2 & \sqrt{3}/2
\end{pmatrix} , B= \begin{pmatrix}
1 \\
1 \\
0
\end{pmatrix},  C=\begin{pmatrix}
1 & 0 & 0
\end{pmatrix}
\] and in the case of full-state coupling $C=I$.
In the following three cases, we simulate the regulated state synchronization via protocols \eqref{pscp1} and \eqref{pscp3}, for full- and partial-state couplings respectively. Moreover, we present plots of the states, exosystem, and relative difference (error) in each case. We note that from \eqref{boundtau}, the given agent model can only admit $\bar{\kappa} < 3$. For all of our partial state examples, we have $ F\T = (2.1321, 0.5469, 1.0299)$.
\begin{figure}[t]
	\includegraphics[width=9cm, height=7.2cm]{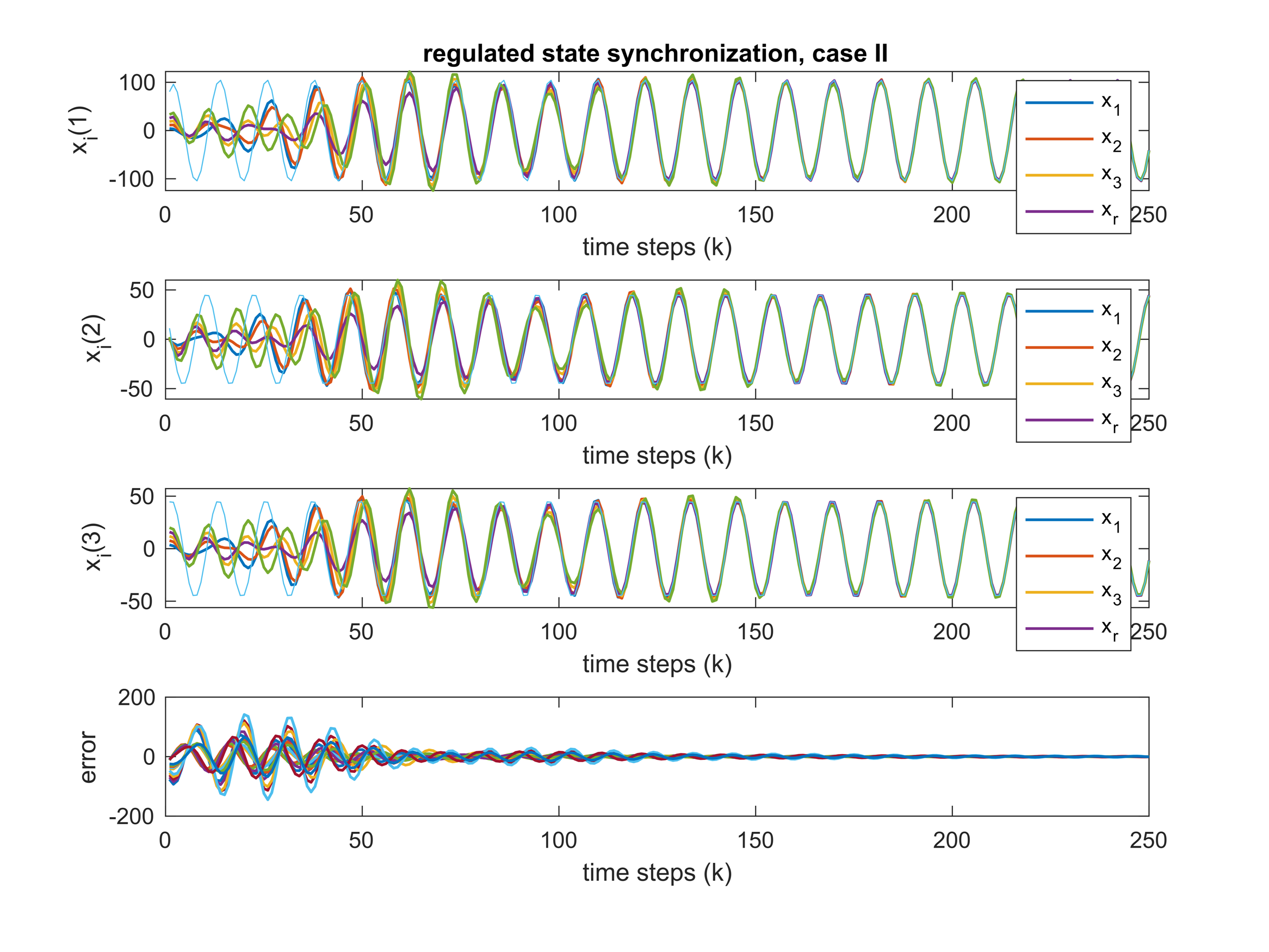}
	\centering
	\vspace{-1cm}
	\caption{Directed communication network with $5$ nodes and full-state coupling.}\label{5nodes_full}
\end{figure}
\begin{figure}[t]
	\includegraphics[width=9cm, height=7.2cm]{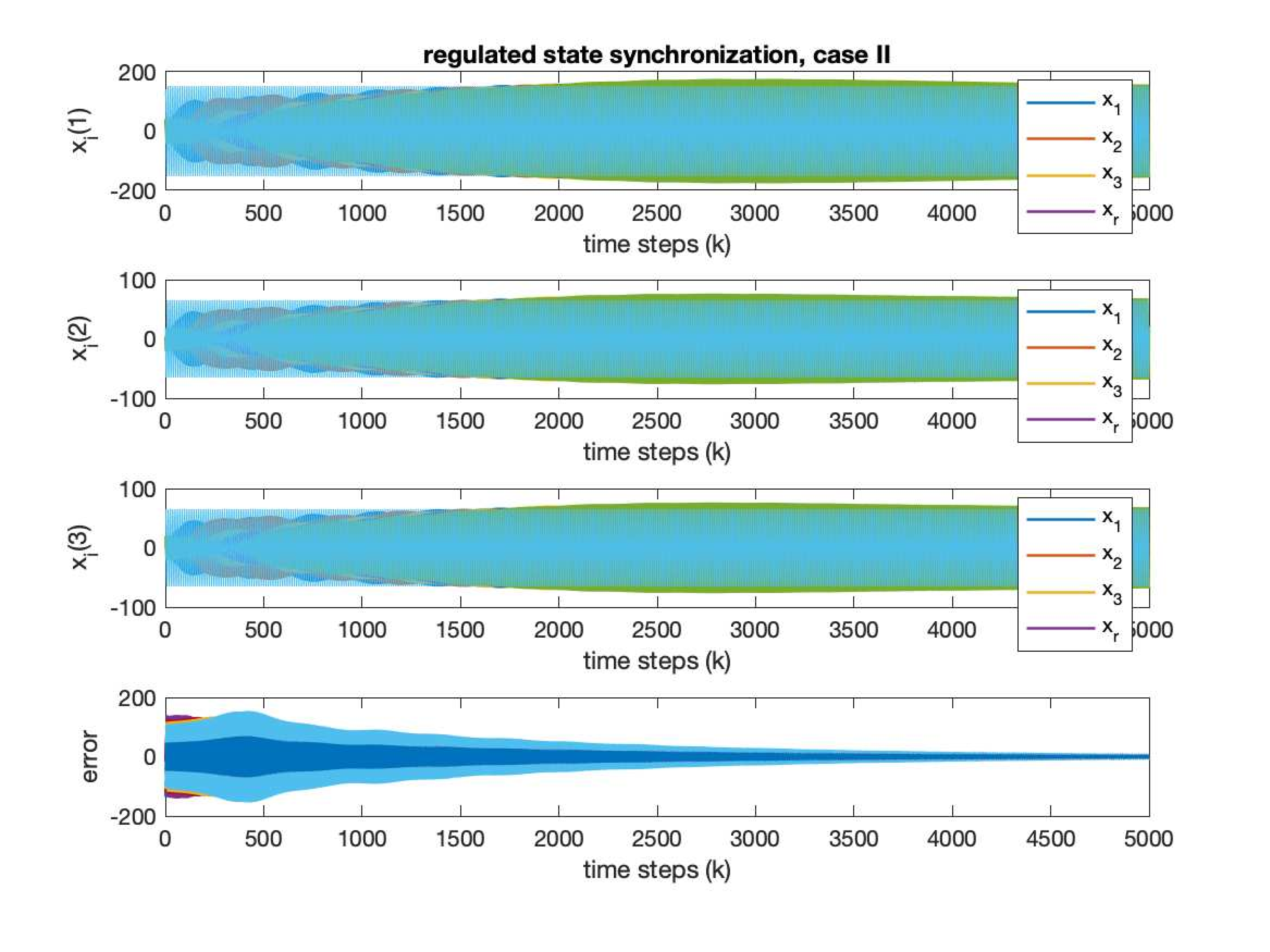}
	\centering
	\vspace{-1cm}
	\caption{Directed communication network with $5$ nodes and partial-state coupling.}\label{5nodes_partial}
\end{figure}
\begin{enumerate}
	\item  Firstly, we consider a MAS with $N = 3$ (i.e. 3 agents), and communication network defined by the adjacency matrix with entries $a_{21} = a_{32} = a_{13} = 1$ and all the rest zero. 
	The delays in both networks with full and partial-state coupling are given as follows: $\kappa_1 =\kappa_2 = 1 $, and $\kappa_3 = 2.$ Furthermore, we choose $\epsilon = 0.001.$ We present the results of our MAS synchronization for the full-state and partial-state coupling cases in Figure \ref{graph_3nodesfull} and Figure \ref{graph_3nodespartial}, respectively.
	\item Next, we consider a MAS with $N = 5$ agents and adjacency matrix $\mathcal{A}$ with entries $a_{i+1, i} = 1$ for $i = 1, \dots, 4,$ as well as $a_{13} = a_{35} = 1$, and all the rest zero. 
	The delays are chosen as $\kappa_1=\kappa_2 =\kappa_3= \kappa_5 = 2 $, and $\kappa_4 = 1$ for both networks with full- and partial-state coupling. Furthermore, we choose $\eps = 0.001$ for the full-state case, and $\eps = 0.00001$ for the partial-state case. We present the MAS synchronization results for both of these cases in Figure \ref{5nodes_full} and Figure \ref{5nodes_partial}, respectively.	
	\item In this case, we consider a MAS with $N=10$ and adjacency matrix $\mathcal{A}$ with entries $a_{i+1, i} = 1$ for $i = 1, \dots, 9,$ as well as $a_{15} = a_{1,10} = a_{5,10} = 1$, and all the rest zero. For the full-state coupling case, the delays are chosen as $\kappa_1 = \kappa_3 = \kappa_4 = \kappa_5 = \kappa_8 = 1 $ and $\kappa_2 = \kappa_6 = \kappa_7 = \kappa_9 = \kappa_{10} = 1$. For the partial state coupling case, we let $\kappa_i = 1$ for $i = 1, \dots, 10.$ Moreover, $\eps = 0.001$ for both cases. We present the results of our MAS synchronization for the full-state and partial-state coupling cases in Figure \ref{10nodes_full} and Figure \ref{10nodes_partial}, respectively.
\end{enumerate}
We now make a few observations based on these examples. The simulation results show that our one-shot protocol designs do not need any knowledge of the communication network and achieve regulated state synchronization for any network with any number of agents. Moreover, we observe that upper bounds on the input delay tolerance only depends on the agent dynamics. 
Also, we observe that in each case, full-state synchronization is achieved faster than partial-state. This is to be expected, as full-state coupling communicates \emph{all} states over the network (as opposed to only part of the states), meaning that the agents are given more information about other agents, allowing them to synchronize faster. Note that in all examples, the exosystem produces an oscillating command signal, which all of the agents eventually synchronize with.
\bibliographystyle{plain}
\bibliography{referenc}
\end{document}